\documentclass[11pt,a4paper]{article}

\usepackage[utf8]{inputenc}
\usepackage{mathtools}
\usepackage{amsfonts}
\usepackage{amssymb}
\usepackage{amsthm}
\usepackage{tikz}
\usetikzlibrary{arrows,automata}
\usepackage{url}

\newtheorem{theorem}{Theorem}
\newtheorem{lemma}[theorem]{Lemma}

\newcommand{\N}{{\mathbb{N}}}

\newcommand{\tuple}[1]{\langle #1 \rangle}
\newcommand{\limrun}[0]{\infty}
\newcommand{\emptyword}[0]{\lambda}
\newcommand{\trans}[1]{\mathchoice{\xrightarrow{#1}}{\xrightarrow{\smash{\lower1pt\hbox{$\scriptstyle #1$}}}}{\text{Error}}{\text{Error}}}
\newcommand{\occ}[2]{|#1|_{#2}}
\newcommand{\prefsel}[0]{\upharpoonright}
\newcommand{\noprefsel}[0]{\upharpoonright\mkern-4mu \upharpoonright}

\begin{document}

\title{Preservation of normality by non-oblivious group selection}
\author{Olivier Carton and Joseph Vandehey}

\date{\today}
\maketitle

\begin{abstract}
  We give two different proofs of the fact that non-oblivious selection via
  regular group sets preserves normality.  Non-oblivious here means that
  whether or not a symbol is selected can depend on the symbol itself. One
  proof relies on the incompressibility of normal sequences, the other on
  the use of augmented dynamical systems.
\end{abstract}

\section{Introduction}
  
An infinite sequence $x=x_1x_2x_3x_4\dots$ over a finite alphabet $A$ is said
to be normal if every finite word appears with the same limiting
frequency in~$x$ as every other finite word of the same length. (Fuller
definitions of terminology used in the introduction will be included in the
next section.)

D.D.~Wall \cite{Wall50} famously showed that if $x=x_1x_2x_3\dots$ is
normal, then $x_kx_{k+\ell}x_{k+2\ell}x_{k+3\ell}\dots$ is normal for any
$k,\ell \in \mathbb{N}$. In other words, selecting along an arithmetic
progression preserves normality. Furstenberg \cite{Furstenberg67}, as part
of his seminal paper on the theory of disjoint systems, gave another proof
of this fact. Kamae \cite{Kamae73} and Weiss \cite{weiss71}, using the
theory of disjoint systems as well, were able to characterize those
sequences of positive integers $i_1<i_2<i_3<\dots$ such that selection
along these sequences also preserves normality: in particular, these are
the deterministic sequences or, equivalently, sequences of Kamae entropy zero.

More generally, many mathematicians have studied (prefix) selection rules.
Let $A^*$ be the set of finite words over $A$ and let $L\subset A^*$. The
sequence obtained by \emph{oblivious selection} of~$x$ by~$L$ is $x \prefsel L
= x_{i_1}x_{i_2}x_{i_3} \cdots$, where $i_1,i_2,i_3,\ldots$ is the
enumeration in increasing order of all the integers~$i$ such that the
prefix $x_1 x_2 \cdots x_{i-1}$ belongs to~$L$.  This selection rule is
called \emph{oblivious} because the symbol~$x_i$ is not included in the
considered prefix.  If $L = A^*1$ is the set of words ending with a~$1$,
the sequence~$x \prefsel L$ is made of all symbols of~$x$ occurring after
a~$1$ in the same order as they occur in~$x$.

The examples above by Wall, Kamae, and Weiss are all examples of oblivious
selection rules, where $L$ consists of all words of certain fixed
lengths. However, far more intricate selection rules are possible. The
following theorem of Agafonov \cite{Agafonov68} states that normality is
preserved by oblivious selection of a regular language.
\begin{theorem}[Agafonov] \label{thm:agafonov}
  If the sequence $x \in A^\N$ is normal and $L \subset A^*$ is
  regular, then $x \prefsel L$ is also normal.
\end{theorem}

A language $L\subset A^*$ is regular if it is accepted by a deterministic
finite automaton. We will speak more on this later.

Kamae and Weiss \cite{KamaeWeiss75} extended Theorem \ref{thm:agafonov}
slightly. Let $L$ be a set of words and let $\sim_L$ be an equivalence
relation given by $u \sim_L v$ if $\{w:uw\in L\}=\{w:vw\in L\}$. If
$L/\sim_L$ is finite, then selection along $L$ preserves normality. In
contrast, Merkle and Reimann \cite{MerkleReimann06} showed that selection
by deterministic one-counter languages or by linear languages need not
preserve normality.  We also mention that suffix selection, where the
selection of a given digit is based of the tail of the sequence after that
digit, has also been considered \cite{BecherCartonHeiber15}.

We can also define the sequence obtained by
\emph{non-oblivious selection} of $x$ by $L$. This is
$x \noprefsel L = x_{i_1}x_{i_2}x_{i_3} \cdots$, where $i_1,i_2,i_3,\ldots$
is the enumeration in increasing order of all the integers~$i$ such that
the prefix $x_1 x_2x_3 \cdots x_i$ \emph{including} $x_i$ belongs to~$L$. 

Non-oblivious selection is more powerful than oblivious selection, as it
can simulate the latter due to the following formula:
\begin{displaymath}
  x \noprefsel LA = x \prefsel L
\end{displaymath}
for any sequence~$x$ and any set~$L$ of words.  Let us recall that
$LA$ is the set of words of the form~$wa$ for $w \in L$ and $a \in A$. On
the other hand, there are weaknesses in non-oblivious selection as well. If
we take $L=A^*1$ again, then $x\noprefsel L$ will consist of nothing but
$1$'s, which will not be normal.

Accordingly, oblivious selection has been studied more than
non-oblivious selection. The second author \cite{Vandehey17a} has a few
(very specific) examples of non-oblivious selections that preserve
normality. In this paper, we present a more general theorem:

\begin{theorem} \label{thm:main}
  If the sequence $x \in A^\N$ is normal and $L \subset A^*$ is a
  regular group set, then $x \noprefsel L$ is also normal.
\end{theorem}

Regular here has the same meaning as in Agafonov's theorem. Saying that $L$ is a
group set implies that in the associated deterministic finite
automaton, any input will permute the states.

We note that if there exists a regular group set~$K$ such that the symmetric difference
$L \triangle K$ is finite, then the non-oblivious selection by~$L$ also
preserves normality.  More generally, if $L$ is accepted by an automaton
such that each recurrent strongly connected component (those components, which, once entered, cannot be
left) is a group automaton, then non-oblivious selection by~$L$ still
preserves normality.

We will prove this result using two distinct methods reflecting the
different styles of the two authors of this paper. 

The first method,
favored by the first author, makes use of the fact that normality can be
defined in terms of incompressibility by deterministic finite automata
\cite{BecherCartonHeiber15}. It follows along the
lines of the proof of Agafonov's theorem presented
in~\cite{BecherCarton18}.  One key ingredient of this proof is the
statement that the function which maps each sequence~$x$ to the pair
$(x \prefsel L, x \prefsel (A^* \setminus L))$ is one-to-one.  The same
statement for the non-oblivious selection does not hold, even when $L$ is a
group set as is shown by the following example.  Consider the set~$L$ of
words accepted by the automaton pictured in Figure~\ref{fig:group-autom}.
Let $x$ and $x'$ be the sequences $01^\mathbb{N}$ and $101^\mathbb{N}$.  It is
easily computed that $x \noprefsel L = x' \noprefsel L = 01^\mathbb{N}$ and
$x \noprefsel (A^* \setminus L) = x' \noprefsel (A^* \setminus L) =
1^\mathbb{N}$.

\begin{figure}[htbp]
  \begin{center}
  \begin{tikzpicture}[->,>=stealth',initial text=,semithick,auto,inner sep=1.5pt]
    \tikzstyle{every state}=[minimum size=0.4]
    \node[state,initial above] (q0) at (0,0) {$q_0$};
    \node[state,accepting]  (q1) at (2,0) {$q_1$};
    \node[state]  (q2) at (4,0) {$q_2$};
    \path (q0) edge[out=210,in=150,loop] node {$1$} ();
    \path (q0) edge[bend left=20] node {$0$} (q1);
    \path (q1) edge[bend left=20] node {$0$} (q0);
    \path (q1) edge[bend left=20] node {$1$} (q2);
    \path (q2) edge[bend left=20] node {$1$} (q1);
    \path (q2) edge[out=30,in=-30,loop] node {$0$} ();
  \end{tikzpicture}
  \end{center}
  \caption{A group automaton with a final state at $q_1$}
  \label{fig:group-autom}
\end{figure}
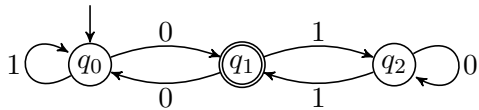

The second method, favored by the second
author, makes use of the
idea of augmented systems, dynamical systems which have been extended to
simultaneously act over a deterministic finite automaton \cite{HeersinkVandehey16,Vandehey17}. In this framework, the distinction between oblivious and non-oblivious selection is much smaller. We make use of a technique from \cite{Vandehey17a}, where we use an automaton that also records a finite number of  selected symbols: this reduces the problem of counting frequencies of words in $x\noprefsel L$ to the problem of calculating visiting frequencies of certain states in the automaton.

In Section \ref{sec:prelim}, we will give definitions, notation, and the necessary results from previous papers. In Section \ref{sec:method1}, we provide a proof of Theorem \ref{thm:main}, following the ideas of incompressibility. In Section \ref{sec:method2}, we provide a proof of Theorem \ref{thm:main}, following the ideas of augmented systems.

\section{Preliminaries} \label{sec:prelim}

\subsection{Sequences, words, and normality}

We write $\N=\{1,2,3,\dots\}$ for the set of all natural numbers.  An
\emph{alphabet}~$A$ is a finite set with at least two symbols.  We
respectively write $A^*$ and $A^{\N}$ for the set of all finite sequences
(also known as \emph{words}) and the set of all infinite sequences of
elements of~$A$ (which we will simply refer to as \emph{sequences}).  We also write $A^k$
stands for the set of all words of
length~$k$.  The length of a finite word~$w$ is denoted by~$|w|$. The empty word is denoted by $\emptyword$.  The
positions in finite and infinite sequences are numbered starting from~$1$.
For a word~$w$ and positions $1 \le i \le j \le |w|$, we let $w[i]$ and
$w[i..j]$ denote the symbol~$a_i$ at position~$i$ and the word
$a_ia_{i+1}\cdots a_j$ from position $i$ to position~$j$.  A set~$C$ of
words is called \emph{prefix-free} if for any $u,v \in C$ with $u$ being a
prefix of $v$ ($u=v[1..|u|]$), we have $u = v$.  We write $\log$ for the
base~$2$ logarithm.  For any finite set $S$ we denote its cardinality with
$\#S$. 
 
Let $x = a_1a_2a_3 \cdots $ be a sequence over the alphabet~$A$.  Let
$L \subseteq A^*$ be a set of words over~$A$.  As described in the introduction, the sequence obtained by
\emph{oblivious selection} of~$x$ by~$L$ is
$x \prefsel L = a_{i_1}a_{i_2}a_{i_3} \cdots$, where $i_1,i_2,i_3,\ldots$
is the enumeration in increasing order of all the integers~$i$ such that
the prefix $a_1 a_2 \cdots a_{i-1}$ belongs to~$L$. The sequence obtained by
\emph{non-oblivious selection} of $x$ by $L$ is
$x \noprefsel L = a_{i_1}a_{i_2}a_{i_3} \cdots$ where $i_1,i_2,i_3,\ldots$
is the enumeration in increasing order of all the integers~$i$ such that
the prefix $a_1 a_2a_3 \cdots a_i$ including $a_i$ belongs to~$L$.

We recall here the notion of normality.  We start with the notation for the
number of occurrences of a given word~$u$ within another word~$w$.  For $w$
and $u$ two words, the number $\occ{w}{u}$ of \emph{occurrences}
of~$u$ in~$w$ is given by $\occ{w}{u} = \#\{ i : w[i..i+|u|-1] = u \}$.
We say that $x\in A^\mathbb{N}$ is \emph{normal}  if for each word
$w \in A^*$, we have
\begin{displaymath}
  \lim_{n \to \infty} \frac{\occ{x[1..n]}{w}}{n} = (\#A)^{-|w|}.
\end{displaymath}
This differs from Borel's original definition~\cite{Borel09} of normality,
but is equivalent (see \cite[Sect.~7.3]{BecherCarton18}).

\subsection{Deterministic finite automata}

A \emph{deterministic finite automaton} is a tuple $\mathcal{T} =
\tuple{Q,A,\delta,I,F}$, where $Q$ is a finite set of \emph{states}, $A$ is
the input alphabet, $\delta: Q \times A \to Q$ is the \emph{transition}
function, and $I \subseteq Q$ and $F \subseteq Q$ are the sets of
\emph{initial} and \emph{final} states, respectively. We focus on automata
that operate in real-time, that is, they process exactly one input alphabet
symbol per transition.  Moreover, we will assume that there is a single
initial state, that is, $I$ is a singleton set.

The relation $\delta(p,a) = q$ is written $p \trans{a} q$ and we further
denote the sequence of consecutive transitions
\begin{displaymath}
  q_0 \trans{a_1} q_1 \trans{a_2} q_2\trans{}
  \cdots \trans{} q_{n-1} \trans{a_n} q_n
\end{displaymath}
by $q_0\trans{a_1a_2\dots a_n} q_n$.  A word $w=a_1a_2\dots a_n$ is said to
be accepted by an automaton if $q_0 \trans{w} q_n$, $q_0$ is the initial
state (that is, $I = \{ q_0 \}$), and $q_n$ is final (that is, in~$F$).  A
language $L\subset A^*$ is said to be regular (as seen in
Theorem~\ref{thm:agafonov}) if there exists a deterministic finite
automaton~$\mathcal{T}$ such that $w\in L$ if and only if $w$ is accepted
by $\mathcal{T}$.

We now introduce a classical class of regular sets called group sets (as
seen in Theorem \ref{thm:main}).  A \emph{group automaton} is a
deterministic automaton such that each symbol induces a permutation of the
states.  By inducing a permutation, we mean that, for each symbol~$a$, the
function which maps each state~$p$ to the state~$q$ such that $p \trans{a}
q$ is a permutation of the state set.  Put another way, if $p \trans{a} q$
and $p' \trans{a} q$ are two transitions of the automaton, then $p = p'$.
A regular set~$L \subseteq A^*$ is called a \emph{group set} if $L$ is
accepted by a group automaton.  It is well known that a regular set is a
group set if and only if its syntactic monoid is a group
\cite[Sect.~7.5]{Pin97}.

\subsection{Dynamical systems}

We will consider dynamical systems to consist of a tuple
$\mathcal{X}=\langle X,\mathcal{F},T,\mu\rangle$, where $X$ is a space, $\mathcal{F}$
is a $\sigma$-algebra on this space, $T:X\to X$ is a transformation (always assumed to be continuous with respect to $\mathcal{F}$), and $\mu$
is a measure on $\mathcal{F}$. We say that $T$ preserves the measure $\mu$,
or, equivalently, that $\mu$ is $T$-invariant, if $\mu(T^{-1}A)=\mu(A)$ for
all $A\in \mathcal{F}$. We say that a system is ergodic if
$\mu(T^{-1}A \triangle A)=0$ for some $A\in \mathcal{F}$ implies that
$\mu(A)=0$ or $\mu(X\setminus A)=0$. Ergodicity may be considered an indecomposability criterion for dynamical systems: a system is ergodic if it cannot be split into two large $T$-invariant pieces.

Of particular relevance to this paper is the symbolic shift system. Let
$X=A^\N$ be the space of all infinite words on the alphabet $A$. Let $T$
denote the forward shift on $X$, so that for a sequence $x\in X$,
$(Tx)_i=x_{i+1}$. Given a word $w\in A^*$, let $C_w$ denote the
cylinder set corresponding to $w$, so that $C_w$ consists of all $x\in X$
such that $x[1..|w|]=w$. The cylinder sets form a semi-algebra that
generates the canonical $\sigma$-algebra on $X$ and so if we let $\mu$ be a
measure on cylinder sets $C_w$ given by $\mu(C_w)=(\#A)^{-|w|}$, then this
extends to a measure on this $\sigma$-algebra. We note that if
$A=\{0,1,\dots,b-1\}$ and we use the standard bijection\footnote{Actually
  not quite a bijection, due to the phenomenon of $0.0\overline{9}=0.1$,
  but this will not be relevant and can be ignored.} from $X$ to $[0,1)$
associating an infinite word with a $b$-ary expansion, then $\mu$ is just the
Lebesgue measure.

We may reinterpret the definition of normality in a slightly more ergodic
manner. Since $x[i..i+|w|-1]=w$ if and only if $T^{i-1}x \in C_w$, we have
that a word $x\in X$ is normal if for each finite word $w\in A^*$, we have
\[
  \lim_{n\to\infty} \frac{\#\{0\le i \le n-1: T^i x\in C_w\}}{n}=\mu(C_w).
\]
Note that $T$ is ergodic and invariant with respect to~$\mu$, so by the
pointwise ergodic theorem, $\mu$-almost all $x\in X$ are normal.

By the Pyatetskii-Shapiro normality criterion \cite{MoshchevitinShkredov03}, we
can weaken the above to say that $x$ is normal if there exists a fixed
constant $c>0$ such that for every $w\in A^*$, we have
\[
  \limsup_{n\to\infty} \frac{\#\{0\le i \le n-1: T^i x\in C_w\}}{n}\le c\mu(C_w).
\]

\section{The first method of proof} \label{sec:method1}

We now introduce automata with output, also known as transducers, which are
used to compress sequences and to select symbols from a sequence.  In this
paper we only consider input-deterministic transducers (also known as
\emph{sequential}) computing functions from sequences to sequences.  Such a
machine is a deterministic automaton in which each transition is equipped
with an additional output word.  The output of a (infinite) run is the
concatenation of the outputs of the transitions used by the run. More
formally a \emph{transducer} is a tuple $\mathcal{T} =
\tuple{Q,A,B,\delta,I,F}$, where $Q$ is a finite set of \emph{states}, $A$
and $B$ are the input and output alphabets, respectively, $\delta: Q \times
A \to B^* \times Q$ is the \emph{transition} function, and $I \subseteq Q$
and $F \subseteq Q$ are the sets of \emph{initial} and \emph{final} states,
respectively.  The set~$I$ is again a singleton set.

The relation $\delta(p,a) = (w,q)$ is written $p \trans{a|w} q$ and the
tuple $\tuple{p,a,w,q}$ is then called a \emph{transition} of the
transducer.  A finite (respectively, infinite) \emph{run} is a finite
(respectively, infinite) sequence of consecutive transitions,
\begin{displaymath}
  q_0 \trans{a_1|v_1} q_1 \trans{a_2|v_2} q_2
  \trans{} \cdots \trans{} q_{n-1} \trans{a_n|v_n} q_n.
\end{displaymath}
Its \emph{input and output labels} are respectively $a_1\cdots a_n$ and
$v_1 \cdots v_n$.  A finite run is written $q_0 \trans{a_1\cdots a_n|v_1
  \cdots v_n} q_n$.  An infinite run is \emph{final} if the state~$q_n$ is
final for infinitely many integers~$n$.  In that case, the infinite run is
written $q_0 \trans{a_1a_2a_3\cdots|v_1v_2v_3\cdots} \limrun$.  An infinite
run is accepting if it is final and furthermore its first state~$q_0$ is
the initial one.  This is the classical B\"uchi acceptance condition
\cite{PerrinPin04}.  Since transducers are supposed to be input-deterministic,
there is at most one accepting run $q_0 \trans{x|y} \limrun$ having a given
sequence~$x$ for input label and we write $y = \mathcal{T}(x)$. 

A transducer is called \emph{one-to-one} if the function which maps $x$
to~$y$ is one-to-one.  We always assume that all transducers are
\emph{trim}: each state occurs in at least one accepting run.

A sequence $x = a_1 a_2 a_3 \cdots$ is \emph{compressible} by a
transducer~$\mathcal{T}$ if it has an accepting run
$q_0 \trans{a_1|v_1} q_1 \trans{a_2|v_2} q_2 \trans{a_3|v_3} q_3 \cdots $
satisfying
\begin{displaymath}
  \liminf_{n \to \infty} \frac{|v_1 v_2 \cdots v_n| \log \#B}{n \log \#A} < 1.
\end{displaymath}
Recall that each of the $v_i$'s belongs to $B^*$, not necessarily $B$, so could be empty or have length greater than $1$.

The connection between compressible sequences and normality is given by the following:

\begin{theorem} \label{thm:compress}
  A sequence is normal if and only if it not compressible by a one-to-one
  deterministic transducer.
\end{theorem}

The above result follows from the results in~\cite{Dai04,SchnorrStimm72}.
A direct proof appears in~\cite{BecherHeiber13}.  Extensions of this
characterization for non-deterministic and extra memory transducers are
in~\cite{BecherCartonHeiber15,CartonHeiber15}.

Let $c : A^k \to A^*$ be a function mapping each word of length~$k$ to some
word.  This function can be extended to a function from $(A^k)^*$ to~$A^*$
by setting $c(w_1 \cdots w_n) = c(w_1) \cdots c(w_n)$ with $w_i \in A^k$ for
$1 \le i \le n$.  When a sequence~$x$ is not normal, it can be compressed
using a Huffman coding.  This is implicit in the following lemma.  The proof
of the next lemma is the first part of the proof of Lemma~7.5.1
in~\cite{BecherCarton18}.
\begin{lemma} \label{lem:huffman}
  If the sequence~$x$ is not normal, there is a length~$k$ and a one-to-one
  function $c : A^k \to C$ where $C$ is a prefix-free set such that
  \begin{displaymath}
    \liminf_{n \to \infty} \frac{|c(x[1..nk])|}{nk} < 1.
  \end{displaymath}
\end{lemma}

The next lemma states that if the input~$x$ is normal, each state
which is visited infinitely often in the run over~$x$ in some deterministic
automaton is visited more than a linear number of times.  This lemma is
actually Lemma~7.10.3 in~\cite{BecherCarton18}.
\begin{lemma} \label{lem:nonzerofreq}
  Let $x = a_0a_1a_2\cdots$ be a normal sequence and let $q_0 \trans{a_0} q_1
  \trans{a_1} q_2 \trans{a_2} \cdots$ be a run in a deterministic
  automaton.  If the state $q$ is visited infinitely often in this run, then $\liminf_{n
    \to \infty} \#\{i \le n : q_i = q \}/n > 0$.
\end{lemma}

A transducer can have two output tapes. In this case, the transducer would be a tuple $\mathcal{T}=\langle Q,A,B,\delta,I,F\rangle$, where now the transition function is $\delta:Q\times A\to B^*\times B^*\times Q$. The relation $\delta(p,a)=(w_1,w_2,q)$ is written $p\trans{a|w_1,w_2}q$.

A deterministic automaton~$\mathcal{A}$ accepting a set~$L$ can can be turned
into a two-output transducer~$\mathcal{T}$ that outputs $x \noprefsel L$
and $x \noprefsel (A^* \setminus L)$ on its first and second output tapes
respectively.  Each transition $p \trans{a} q$ of~$\mathcal{A}$ is replaced
by either the transition $p \trans{a|a,\emptyword} q$ if the state~$q$ is
final or by the transition $p \trans{a|\emptyword,a} q$ if $q$ is not
final. If $q$ is final, any finite run in~$\mathcal{A}$ from the initial
state to~$q$ is accepting and therefore the label~$a$ of its last
transition must be output to the first output tape because it is selected
in $x \noprefsel L$.  The following lemma states the key property of this
transducer~$\mathcal{T}$ when $\mathcal{A}$ is a group automaton.

\begin{lemma} \label{lem:onetoone}
  Let $\mathcal{A}$ be a group automaton and let $\mathcal{T}$ the
  transducer obtained from~$\mathcal{A}$ as above.  The function
  which maps each finite run $p \trans{u|v,w} q$ of~$\mathcal{T}$ to
  the triple $\tuple{q,v,w}$ is one-to-one.
\end{lemma}
\begin{proof}
  Note first that $|u| = |v| + |w|$.  The proof is done by induction on the
  sum $|v| + |w|$. We will, throughout the proof, denote the run $p \trans{u|v,w} q$ by $\rho$.
  
  If $v$ and~$w$ are empty, $u$ is also empty and so
  $\rho$ must be the empty run from $p = q$ to~$p$.  Suppose now that
  $|v| + |w| > 0$.  Suppose that $q$ is final in the
  automaton~$\mathcal{A}$.  All transitions ending in~$q$ have the form
  $r \trans{a|a,\emptyword} q$ for some state~$r$ and some symbol~$a$.  It
  follows that $v$ cannot be empty.  Let $a$ be the last symbol of~$v$ and
  let $v'$ be such that $v = v'a$.  Since $\mathcal{A}$ is a group
  automaton, there is exactly one transition $r \trans{a} q$ ending in~$q$
  and having $a$ for label.  This implies that the last transition
  of~$\rho$ is $r \trans{a|a,\emptyword} q$. Applying the induction hypothesis to the triple
  $\tuple{r,v',w}$ completes the proof in this case. The case where $q$ is not
  final in~$\mathcal{A}$ works similarly.
\end{proof}

\begin{proof}[Proof of Theorem~\ref{thm:main}]
  Let $x$ be a normal word.  Let $L \subset A^*$ be a regular group set.
  We suppose that $x \noprefsel L$ is not normal and will show
  that $x$ can be compressed by a one-to-one deterministic transducer, contradicting its normality.
  
  First note that $x \noprefsel (A^* \setminus L)$ cannot be a finite word,
  as otherwise the normality of $x\noprefsel L$ is trivial.
  
  Let $\mathcal{A}$ be a group automaton accepting~$L$ whose state set and
  transition function are $Q$ and~$\delta$ respectively.  Let
  $\mathcal{T}_1$ be the two-output transducer obtained from~$\mathcal{A}$
  as above. The state set of~$\mathcal{T}_1$ is the same as the one
  of~$\mathcal{A}$ and the transition function~$\delta_1$
  of~$\mathcal{T}_1$ is defined as follows.
  \begin{displaymath}
    \delta_1(p,a) =
    \begin{cases}
      (a,\emptyword,\delta(p,a)) & \text{if }\delta(p,a) \in F \\
      (\emptyword,a,\delta(p,a)) & \text{if }\delta(p,a) \notin F.
    \end{cases}
  \end{displaymath}

  Since this transducer is assumed to be trim and since $x$ is assumed to
  be normal, it must reach final states infinitely often, showing that
  $x\noprefsel L$ is infinite. This is a consequence of
  \cite[Satz~2.5]{SchnorrStimm72}.
  
  Since it is supposed that $y = x \noprefsel L$ is not normal, there is,
  by Lemma~\ref{lem:huffman}, an integer~$k$ and a one-to-one function~$c$
  from $A^k$ into a prefix-free set~$C$ such that
  \begin{displaymath}
    \liminf_{n \to \infty} \frac{|c(y[1..nk])|}{nk} < 1.
  \end{displaymath}
  The transducer~$\mathcal{T}_1$ can be combined with the function~$c$ to
  get a new transducer~$\mathcal{T}_2$ which writes $c(y[1..nk])$ instead
  of $y[1..nk]$ on its first output tape.  This latter transducer has a
  buffer $\mathcal{B} = \mathcal{B}(k)$ of size~$k$ in which each symbol
  of~$y$ is put.  Whenever this buffer becomes full containing a word~$w$,
  the transducer~$\mathcal{T}_2$ writes $c(w)$ to its first output tape.
  $\mathcal{T}_2$ also behaves identically to $\mathcal{T}_1$ on its second
  output tape.  To be more precise, the state set of~$\mathcal{T}_2$ is $Q
  \times \mathcal{B}(k)$ where $\mathcal{B}(k) = \bigcup_{i=0}^{k-1} B^i$.
  The transition function~$\delta_2$ of~$\mathcal{T}_2$ is defined as
  follows.

  \begin{displaymath}
    \delta_2((p,w),a)=
    \begin{cases}
      (\emptyword,\emptyword, (\delta(p,a),wa))
        & \text{if }\delta(p,a)\in F \text{ and } |wa|<k,\\
      (c(wa),\emptyword,(\delta(p,a),\emptyword)),
        & \text{if }\delta(p,a)\in  F\text{ and } |wa|=k, \\
      (\emptyword,a,(\delta(p,a),w)), & \text{if }\delta(p,a) \notin F.
    \end{cases}
  \end{displaymath} 
  We claim that $\mathcal{T}_2$ compresses its input~$x$.  By this, we mean
  that if $u_n$ and~$v_n$ are the outputs of~$\mathcal{T}_2$ on its two
  tapes after $n$ transitions (consuming $n$ input symbols), then
  $\liminf_{n \to \infty}{(|u_n|+|v_n|)/n} < 1$.  The result is clear
  because the coding function~$c$ compresses the word~$y$ and because, by
  Lemma \ref{lem:nonzerofreq}, along the run of the
  automaton~$\mathcal{A}$, final states are visited at linearly many times.
  
  We now construct a new transducer~$\mathcal{T}_3$ which merges the
  contents of the two output tapes of~$\mathcal{T}_2$ into a single output.
  Let $m$ be a integer to be fixed later.  The transducer~$\mathcal{T}_3$
  has one output tape and two buffers $\mathcal{B}_1 = \mathcal{B}(m)$
  and~$\mathcal{B}_2 = \mathcal{B}(m)$ of size~$m$, constructed in a
  similar way to $\mathcal{T}_2$ above.  Whenever $\mathcal{T}_2$ writes a
  symbol to its first (respectively, second) output tape, this symbol is
  added to the buffer~$\mathcal{B}_1$ (respectively, $\mathcal{B}_2$).
  Whenever one buffer $\mathcal{B}_1$ or~$\mathcal{B}_2$ becomes full, its
  content is written to the output tape of~$\mathcal{T}_3$ with an extra
  bit in front of it to indicate whether the content comes from either
  $\mathcal{B}_1$ or~$\mathcal{B}_2$.  This extra bit is~$0$ if the content
  comes from~$\mathcal{B}_1$ and $1$ otherwise.  To be more precise, the
  state set of~$\mathcal{T}_3$ is $Q \times \mathcal{B}(k) \times
  \mathcal{B}(m) \times \mathcal{B}(m)$ and its transition
  function~$\delta_3$ is defined as follows.
  \begin{displaymath}
    \delta_3((p,w_1,w_2,w_3),a)=
    \begin{cases}
      (\emptyword,(\delta(p,a),w_1a,w_2,w_3)) \\
      \qquad \text{ if } \delta(p,a) \in F , |w_1a| < k \\
      (\emptyword,(\delta(p,a),\emptyword,w_2c(w_1a),w_3)) \\
      \qquad \text{ if } \delta(p,a) \in F , |w_1a| = k \\
      \qquad \quad |w_2c(w_1a)| < m \\
      (0u,(\delta(p,a),\emptyword,v,w_3)) \\
      \qquad \text{ if } \delta(p,a) \in F , |w_1a| = k \\
      \qquad \quad w_2c(w_1a) = uv \text{ where } |u| = m \\
      (\emptyword,(\delta(p,a),w_1,w_2,w_3a)) \\
      \qquad \text{ if } \delta(p,a) \notin F , |w_3a| < m \\
      (1w_3a,(\delta(p,a),w_1,w_2,\emptyword)) \\
      \qquad \text{ if } \delta(p,a) \notin F , |w_3a| = m \\
    \end{cases}
  \end{displaymath}
  We assume that $m$ is sufficiently large so that $|c(w)|\le m$ for any
  $w\in A^k$: this guarantees that in the third case above, we have
  $m\le w_2c(w_1a)<2m$ so $v$ is in $\mathcal{B}(m)$ as desired.  Moreover,
  we claim that, for $m$ great enough, the transducer~$\mathcal{T}_3$ also
  compresses its input~$x$.  Note that the output of~$\mathcal{T}_3$ is
  longer than the sum of the two outputs of~$\mathcal{T}_2$ because each
  block of~$m$ symbols is preceded by an extra bit $0$ or~$1$.  However,
  for $m$ great enough, this loss is offset by the compression
  of~$\mathcal{T}_2$, and indeed $\mathcal{T}_3$ compresses its input.
  
  Note that none of the transducers $\mathcal{T}_1$, $\mathcal{T}_2$, and
  $\mathcal{T}_3$ is one-to-one because the function which maps $x$ to the
  pair $(x \noprefsel L, x \noprefsel (A^* \setminus L))$ might not be
  one-to-one.  For that reason, we construct a last
  transducer~$\mathcal{T}_4$ obtained by changing $\mathcal{T}_3$ to make
  it one-to-one.  The transducer~$\mathcal{T}_4$ works as $\mathcal{T}_3$
  but whenever $\mathcal{T}_3$ writes to its output tape a block of
  length~$m$ coming from its buffer~$\mathcal{B}_1$ with its extra bit~$0$,
  the transducer~$\mathcal{T}_4$ also writes some extra information that we
  now describe.  This extra information is made of two data.  The first one
  is the current state $\delta(p,a)$ of the automaton~$\mathcal{A}$.  This
  other one is the length of the buffer~$\mathcal{B}_2$.  Both data are
  written in binary and require $\lceil \log \#Q \rceil$ and $ \lceil \log
  m \rceil$ bits respectively.  We do not give explicitly the transition
  function of~$\mathcal{T}_4$ as it is almost the same as the one
  of~$\mathcal{T}_3$.

  Since the additional information written by~$\mathcal{T}_4$ is $o(m)$,
  $\mathcal{T}_4$ still compresses its input for $m$ large enough.  We also
  claim that the transducer~$\mathcal{T}_4$ is one-to-one.  To do this, we
  show that, from the output of $\mathcal{T}_4$, it is possible to recover
  the input.
  
  Consider any block in the output which comes from $\mathcal{B}_1$
  (indicated by the first binary bit in front of it). If we take this block
  and all preceding blocks that arise from $\mathcal{B}_1$, apply $c^{-1}$
  to them, and concatenate them in order, we obtain a prefix $v$ of
  $x\noprefsel L$. In a similar way, if we take all preceding blocks that
  arise from $\mathcal{B}_2$ and the numbers of symbols of the next
  $\mathcal{B}_2$ block indicated by the second data, we obtain a prefix
  $w$ of $x\noprefsel (A^* \setminus L)$. (Recall our assumption that
  $x\noprefsel(A^*\setminus L)$ is infinite, so this next block must always
  exist.) Finally, the first data gives us a state~$q$ for our original
  automaton $\mathcal{A}$. From the way that $\mathcal{T}_3$ is
  constructed, this specific triple $\langle q,v,w\rangle$, represents the
  outputs and state reached by a run $q_0\trans{u|v,w} q$. By
  Lemma~\ref{lem:onetoone}, the triple $\langle q,v,w\rangle$ uniquely
  identifies this $u$ and means it must necessarily be a prefix of $x$.
  Since there are infinitely many blocks that come from $\mathcal{B}_1$,
  this gives us infinitely many prefixes of $x$ and so we know the entirety
  of $x$.
  
  
  
\end{proof}

\section{The second method of proof} \label{sec:method2}

We now consider what happens when we augment a dynamical system to
simultaneously run over a finite state automaton. Let us consider, as
before, a dynamical system $\langle X,\mathcal{F}, T,\mu\rangle $ and an automaton $\langle Q, A, \delta, I, F\rangle$. We will say this automaton is
transitive if given any $q_1,q_2\in Q$ there exists a word $w\in A^*$ such
that $q_1\trans{w} q_2$.

We now consider the following augmented dynamical system $\langle \widetilde{X},\widetilde{\mathcal{F}},\widetilde{T},\tilde{\mu}\rangle$:
\begin{enumerate}
\item $\widetilde{X}:=X\times Q$,
\item $\widetilde{T}:\widetilde{X}\to \widetilde{X}$ given by
  $\widetilde{T}(x,q)=(Tx,\delta(q,x_1))$, where $x=x_1x_2x_3\dots$,
\item Cylinder sets $C_{w,q}=C_w\times \{q\}$ for $w\in A^*$, $q\in Q$
  (noting again that these cylinder sets generate the $\sigma$-algebra
  $\widetilde{\mathcal{F}}$ on $\widetilde{X}$),
\item $\tilde{\mu}(C_{w,q})=(\#A)^{-|w|} / (\#Q)$.
\end{enumerate}

We extend our definition of normality on augmented systems:
$(x,p)\in\widetilde{X}$ is said to be normal  if for every $w\in A^*$ and $q\in Q$, we have
\[
\lim_{n\to\infty} \frac{\#\{0\le i \le n-1: T^i (x,p)\in C_{w,q}\}}{n}=\tilde\mu(C_{w,q}).
\]

We have then the following result:
\begin{theorem}\label{thm:augment}
  Suppose that the automaton $\langle Q,A,\delta,F\rangle$ is a transitive
  automaton. If $\widetilde{T}$ preserves the measure $\tilde\mu$, then $\widetilde{T}$ is ergodic.  Moreover, for any $x\in X$
  that is normal, the point $(x,q)$ is normal w.r.t $\tilde\mu$ for any
  $q\in Q$.
\end{theorem}
This is a simplified verison of Theorem 3.1 in \cite{Vandehey17} (see
Remark 3.2 in that paper for the discussion of the necessary conditions
needed on the dynamic system and note that they are all trivial in our
case). See also \cite{Vandehey14} for a simpler proof.

We again want to consider adding a buffer to an automaton; however, unlike
in the previous section, we wish to consider a ``rolling" buffer, which
will continuously record the previous few inputs that caused us to reach a
final state without resetting itself to a shorter word.

For a given $k\in \mathbb{N}$ and automaton
$\langle Q, A, \delta, I, F\rangle$ as above, consider the automaton
\[
\left\langle Q_k=Q\times A^k, A, \delta_k,I_k, F_k=F\times A^k\right\rangle,
\]
where $\delta_k$ satisfies the following rules:
\begin{itemize}
\item If $\delta(q,a)$ is not final, then
  $\delta_k((q,w),a)=(\delta(q,a),w)$.
\item If $\delta(q,a)$ is final, then
  $\delta_k((q,w),a)=(\delta(q,a),w[2..k-1]a)$,
\end{itemize}
and $I_k=\{(q_0,w_0)\}$, where $w_0$ is any element of $A^k$. The choice of
which $w_0$ to use will not be relevant for any subsequent proofs.  We will
refer this new automaton as a $k$-digit buffer over the original automaton.

\begin{lemma}\label{lemma:transitivepiece}
  Let $x\in X$ be normal and $q\in Q_k$, and let
  $\langle Q_k, A,\delta_k,I_k,F_k\rangle$ be a $k$-digit buffer over a
  transitive group automaton. Then there is a subset $Q'\subseteq Q_k$ such
  that $F_k\cap Q'$ is non-empty, $\langle Q',A,\delta_k,I_k\cap Q',F_k\cap Q'\rangle$ is a transitive automaton
 , and $\widetilde{T}^i(x,q)$ will eventually always be in $Q'$ in its
  second coordinate.  Moreover, every word in $A^k$ will appear in the
  second coordinate of some element of $Q'$.
\end{lemma}

We are being somewhat imprecise about the initial states in the new automaton $\langle Q',A,\delta_k,I_k\cap Q',F_k\cap Q'\rangle$. It is possible that $I_k\cap Q'=\emptyset$. In this case we would replace $I_k\cap Q'$ with the first state in $Q'$ that the $\widetilde{T}$-orbit of $(x,q)$ enters---in essence, shifting everything forward.

\begin{proof}
  Most of this is proved in Lemma 4.1 of
  \cite{Vandehey17} with the exception the non-emptiness of $F_k\cap Q'$ and the last line. 
  
  Consider any $(q_1,w)\in Q'$ and $q_2\in Q$. Since $\langle Q,A,\delta,I,F\rangle$ is transitive, there exists another word, $u\in A^*$, so that $q_1\trans{u}q_2$ in this automaton. Moreover, since $\langle Q_k, A,\delta_k,I_k,F_k\rangle$ behaves the same as  $\langle Q,A,\delta,I,F\rangle$ in the first coordinate, we must have that $(q_1,w) \trans{u}(q_2,w_2)$ for some word $w_2\in A^k$. Since $q_2$ is an arbitrary element of $Q$, we have that every element of $Q$ appears in the first coordinate of $Q'$. In particular, $F_k\cap Q'$ is non-empty.

   We will now show that every word in $A^k$ will appear in the second coordinate of some element of $Q'$. 
   Note that if one starts at a state $(q,w)\in F_k\cap Q'$, it is always possible to reach the
  next final state in $F_k\cap Q'$ via any element in $A$. In particular, 
  $\langle Q,A,\delta,I,F\rangle$ itself is a group automaton and so the
  action of any input is to permute the states. Thus if one starts at $q$
  and keeps repeating the input $a_1$, one must eventually arrive at state $q_1\in F$. At worst the permutation that $a_1$ induces on the states $Q$ has $q$ being the only final state in its cycle. But even in this case, we would just have $q_1=q$. Thus, in our buffered automaton, we have that by inputting $a_1$ enough times we will move from $(q,w)$ to $(q_1,w[2..k]a_1)\in Q'$. We may repeat this process by inputting $a_2$ over and over until we reach a state $(q_2,w[3..k]a_1a_2)\in Q'$, and so on, until we reach $(q_k,a_1a_2\dots a_k)\in Q'$. But since the $a_i$'s are all arbitrary, one can force any desired word to appear in the second
  coordinate of $Q'$.
\end{proof}

\begin{lemma}\label{lemma:preservation}
  Suppose we augment $(X,T,\mu)$ with
  $\langle Q',A,\delta_k,I_k,F_k\rangle$ as defined in the previous
  lemma. Then $\widetilde{T}$ preserves the measure $\tilde{\mu}$.
\end{lemma}

We will use a similar method to the proofs seen in
\cite{Vandehey17a}.

\begin{proof}
  Any set $E\subset \widetilde{X}$ can be decomposed as as
  $E=\bigcup_{(q,w)\in Q'} E_{q,w}\times \{(q,w)\}$. Since the sets
  $\widetilde{T}^{-1}\left( E_{q,w}\times\{(q,w)\}\right)$ are all
  disjoint, if we can show that $\widetilde{T}$ preserves the
  $\tilde\mu$-measure of sets of the form $E_{q,w}\times \{(q,w)\}$ then it
  will follow that $\widetilde{T}$ preserves the $\tilde\mu$-measure of
  $E$, and we are done.

  To prove this, consider the inverse branches $T^{-1}_a$, $a\in A$, of
  $T$, such that $T^{-1}_a x=ax$. We may then likewise decompose
  $\widetilde{T}^{-1}$ into branches $\widetilde{T}^{-1}_a$, $a\in A$, such
  that $\widetilde{T}^{-1}_a$ induces the branch $T^{-1}_a$ in the first
  coordinate. We may then analyze exactly how $\widetilde{T}^{-1}_a$ acts:
  in particular, $\widetilde{T}^{-1}_a(x,(q,w))$ equals
\[
\begin{cases}
(ax,(\delta^{-1}(q,a),w)), & \text{ if } q\not\in F,\\
(ax,(\delta^{-1}(q,a),Aw[1..k-1])), & \text{ if } q\in F \text{ and }w[k]=a\\
\emptyset, & \text{ if } q\in F \text{ and }w[k]\neq a
\end{cases}
\]
where $\delta^{-1}(q,a)$ is the unique state $p\in Q$ such that $\delta(p,a)=q$. (Unique due to $\langle Q,A,\delta,F\rangle$ being a group automaton.)

Since $\mu(aE_{q,w})=\mu(E_{q,w})/(\#A)$ and since for any subset $Y\subset X$ we have that $\tilde\mu(Y\times \{(q,w)\})=\mu(Y)/(\#Q')$, we therefore have that $\tilde{\mu}(\widetilde{T}^{-1}_a( E_{q,w}\times \{(q,w)\}))$ equals
\[
\begin{cases}
\frac{\mu(E_{q,w})}{\#A}\times \frac{1}{\#Q'}, & \text{ if }q\not\in F,\\
\frac{\mu(E_{q,w})}{\#A}\times \frac{\#A}{\#Q'}, & \text{ if } q\in F \text{ and }w[k]=a,\\
0, & \text{ if } q\in F \text{ and }w[k]\neq a.
\end{cases}
\]
By summing over all the $a$'s, we see that 
\[
\tilde\mu(\widetilde{T}^{-1}( E_{q,w}\times \{(q,w)\})) = \frac{\mu(E_{q,w})}{\#A}\times \frac{\#A}{\#Q'}= \frac{\mu(E_{q,w})}{\#Q'} = \tilde\mu(E_{q,w}\times \{(q,w)\})
\]
in all cases, which completes the proof.
\end{proof}

\begin{proof}[Second proof of Theorem~\ref{thm:main}]
 Let $x=a_1a_2a_3\dots \in X$ be normal and let $y=x\noprefsel L=b_1b_2b_3\dots$. Let $\langle Q,A,\delta,I,F\rangle$ denote a group automaton that accepts $L$.
 
 Consider any finite word $w\in A^*$ with length $k=|w|$. By the Pyatetskii-Shapiro normality criterion, we want to show that there exists a uniform $c>0$ (independent of our choice of $w$) such that
 \[
 \limsup_{m\to \infty} \frac{\#\{0\le i \le m-1: T^i y \in C_w\}}{m}\le c\mu(C_w).
 \]
 
 We can analyze how often $w$ appears in $y$ by analyzing the behavior of $x$ when lifted to a augmented system with the $k$-digit buffer $\langle Q_k, A,\delta_k, I_k,F_k\rangle$. The particular lift we choose is $\tilde{x}=(x,(q_0,w_0))$. Let $i_1,i_2,i_3,\dots$ be the increasing sequence of indices $i$ such that $\widetilde{T}^i \tilde{x} \in X\times F_k$, \emph{starting with the $k$th such index}. Then if we let $\pi:\widetilde{X}\to A^k$ be the projection onto the length-$k$ word contained in the second coordinate, then we see that $y[j..j+k-1]=\pi(\widetilde{T}^{i_j}\tilde{x})$.
 
 By Lemma \ref{lemma:transitivepiece}, we know there is a subset $Q'_k\subseteq Q_k$ such that eventually the orbit $\widetilde{T}^i \tilde{x}$ will always be in $X\times Q'_k$. By, as necessary, ignoring a finite piece of $x$, we may assume that we are always in $X\times Q'_k$. Since by Lemma \ref{lemma:preservation}, $\widetilde{T}$ preserves the measure $\tilde{\mu}$, it must also preserve the measure of $\tilde{\mu}$ when restricted to $X\times Q'_k$. Therefore by Theorem \ref{thm:augment}, $\widetilde{T}$ restricted to $X\times Q'_k$ is ergodic and invariant with respect to the restriction of $\tilde{\mu}$.
 
 Let $\widetilde{C}_w$ denote the subset of $Q'_k\cap F_k$ such that the second coordinate is $w$. By the last part of Lemma \ref{lemma:transitivepiece}, this is always non-empty.
 
 Then, with this definition, we have that 
 \begin{align*}
 &\limsup_{m\to \infty} \frac{\#\{0\le i \le m-1: T^i y \in C_w\}}{m}\\ &\qquad = \limsup_{m\to \infty} \frac{\#\{0\le i \le i_m-1: \widetilde{T}^i \tilde{x} \in X\times \widetilde{C}_w\}}{m}
 \\ &\qquad = \limsup_{m\to \infty} \frac{\#\{0\le i \le i_m-1: \widetilde{T}^i \tilde{x} \in X\times \widetilde{C}_w\}}{i_m}\cdot \left(\frac{m}{i_m}\right)^{-1}.
 \end{align*}
 We specified $X\times \widetilde{C}_w$ rather than $\pi^{-1}(w)$ so that the numerator is forced to only count among those indices $i_j$ rather than among all indices $i$.
 
According to Theorem \ref{thm:augment}, $\tilde{x}$ is normal with respect to the restriction of $\widetilde{T}$ to $X\times Q'_k$. Therefore as $m$ tends to infinity we have that
\[
\frac{\#\{0\le i \le i_m-1: \widetilde{T}^i \tilde{x} \in X\times \widetilde{C}_w\}}{i_m}
\]
converges to $\tilde\mu(X\times \widetilde{C}_w)/\tilde\mu(X\times Q'_k)$. Likewise, since we can write $m$ as $\#\{0\le i \le i_m-1: \widetilde{T}^i \tilde{x}\in X\times (Q'_k\cap F_k)\}$, we have that as $m$ tends to infinity, $m/i_m$ converges to $\tilde\mu(X\times (Q'_k\cap F_k))/\tilde\mu(X\times Q'_k)$. Thus,
\[
\limsup_{m\to \infty} \frac{\#\{0\le i \le m-1: T^i y \in C_w\}}{m} = \frac{\tilde\mu(X\times \widetilde{C}_w)}{\tilde\mu(X\times (Q'_k\cap F_k))}.
\]
But, by construction, for any set $E\subset Q_k$, we have that $\tilde\mu(X\times E)=\#E/\#Q_k$. Thus, this limsup is equal to
\[
\frac{\#\widetilde{C}_w}{\#(Q'_k\cap F_k)}.
\]
We want to obtain a crude upper bound on this last fraction. We know that $\widetilde{C}_w\subseteq F\times \{w\}$, so $\#\widetilde{C}_w\le \#F$. Moreover, by the final part of Lemma \ref{lemma:transitivepiece}, we know that $\#(Q'_k\cap F_k)\ge b^k$. Thus,
\[
\limsup_{m\to \infty} \frac{\#\{0\le i \le m-1: T^i y \in C_w\}}{m}\le \frac{\#F}{b^k}.
\]
Setting $c=\#F$ completes the proof.
\end{proof}

\end{document}